\documentclass[lettersize, journal, 10pt]{IEEEtran}
\usepackage{amsmath,amsfonts,amssymb}
\usepackage{array}
\usepackage[caption=false,font=normalsize,labelfont=sf,textfont=sf]{subfig}
\usepackage{textcomp}
\usepackage{stfloats}
\usepackage{amsmath}
\usepackage{url}
\usepackage{verbatim}
\usepackage{graphicx}
\usepackage{cite}
\usepackage{amsthm}
\usepackage{enumitem}
\usepackage{xcolor}
\newtheorem{thm}{Theorem}
\newtheorem{lem}{Lemma}

\newtheorem{assume}{Assumption}
\newtheorem{cor}{Corollary}
\usepackage{tabularx}
\usepackage{adjustbox}
\usepackage{comment}
\usepackage{algpseudocode}
\usepackage{hyperref}

\def\BibTeX{{\rm B\kern-.05em{\sc i\kern-.025em b}\kern-.08em
    T\kern-.1667em\lower.7ex\hbox{E}\kern-.125emX}}

\begin{document}

\title{Goal-Oriented Random Access (GORA)}

\author{Ahsen Topbas, Cagri Ari, Onur Kaya,~\IEEEmembership{Member,~IEEE} and Elif Uysal,~\IEEEmembership{Fellow,~IEEE}
\thanks{A. Topbas, C. Ari, and E. Uysal are with \href{https://cng-eee.metu.edu.tr/}{Communication Networks Research Group (CNG)} at the Department of Electrical and Electronics Engineering, Middle East Technical University (METU), 06800 Ankara, Turkiye (e-mail: \{ahsen.topbas, ari.cagri, uelif\}@metu.edu.tr). \par
A. Topbas and C. Ari are also with Turk Telekom, 06103 Ankara, Turkiye (e-mail: \{ahsen.topbastanyeri, Cagri.ari\}@turktelekom.com.tr). \par
O. Kaya is with the Department of Electrical and Electronics Engineering, Isik University, 34980 Istanbul, Turkiye (e-mail: onur.kaya@isikun.edu.tr).}%
\thanks{This work was supported in part by the TÜBİTAK 1515 Frontier Research and Development Laboratories Support Program for
Turk Telekom neXt Generation Technologies Lab (XGeNTT) under Project
5249902, and in part by the European Union through ERC Advanced Grant 101122990-GO SPACE-ERC-2023-AdG. Views and opinions expressed are, however, those of the authors only and do not necessarily reflect those of the European Union or the European Research Council Executive Agency. Neither the European Union nor the granting authority can be held responsible for them.}
}



\maketitle

\begin{abstract}
We propose \textit{Goal-Oriented Random Access} (GORA), where transmitters jointly optimize what to send and when to access the shared channel to a common access point, considering the ultimate goal of the information transfer at its final destination. This goal is captured by an objective function, which is expressed as a \textit{general} (not necessarily monotonic) function of the Age of Information. Our findings reveal that, under certain conditions, it may be desirable for transmitters to delay channel access intentionally and, when accessing the channel, transmit aged samples to reach a specific goal at the receiver.
\end{abstract}


\vspace{-0.5em}
\section{Introduction}
{
\IEEEPARstart{S}{calable} radio access is key to enabling intelligent applications that rely on data collected from distributed sensors using wireless modalities such as LPWANs and mMTC. \textit{Goal-oriented communication} has recently emerged as a new approach for providing the necessary scalability. This approach represents a paradigm shift, moving beyond the traditional data transmission problem—which primarily focuses on reliability and throughput—to prioritizing \textit{effective communication} aimed at accomplishing specific tasks at the destination~\cite{uysal2022semantic, uysal2024goal}.

Random access is often the radio access method of choice for massive, low-power, and bursty data transmitted by sensors. A number of recent studies have investigated a move toward effective communication in random access~\cite{atabay2020improving, yavascan2021analysis, cavalagli2024reinforcement, chen2020infocom, cocco2023state, wu2024age} by incorporating timeliness. For instance, \cite{atabay2020improving} and \cite{yavascan2021analysis} focused on improving information freshness using the metric \textit{Age of Information} (AoI), which is defined as the time elapsed since the generation of the most recently received sample~\cite{kaul2012real}. The Threshold ALOHA strategy introduced in \cite{atabay2020improving}, controls age by letting users become active based on an age threshold, adopting the \textit{generate-at-will} strategy~\cite{yates2015lazy, sun2017update}. A different approach to effective communication in a random access channel appears in \cite{cocco2023state} for monitoring two-state Markov sources to minimize the state estimation entropy of the sources.

A recent study~\cite{shisher2024timely} has shown that, perhaps counter-intuitively, the reception of aged samples rather than the freshest possible ones, is more useful to the receiver for certain real-time estimation and control applications. For instance, the leader-follower robot setup in \cite[Fig. 2]{shisher2024timely} illustrates that, due to communication delays between the robots, the freshest sample is not the most useful for predicting the current state of the follower robot. Additionally, \cite{ornee2023milcom} has shown that the cost in a safety-critical monitoring system can be a non-monotonic function of AoI. Therefore, it is essential to develop a novel random access strategy that aligns with the principles of goal-oriented communication by integrating data generation and transmission processes. Such a strategy discards the traditional assumption of exogenous data arrivals, allowing nodes to selectively determine which packets to transmit to achieve the desired goal at the destination. Here is an outline of the contributions of this paper:

\begin{figure}
    \centering
\includegraphics[width=0.46\textwidth]{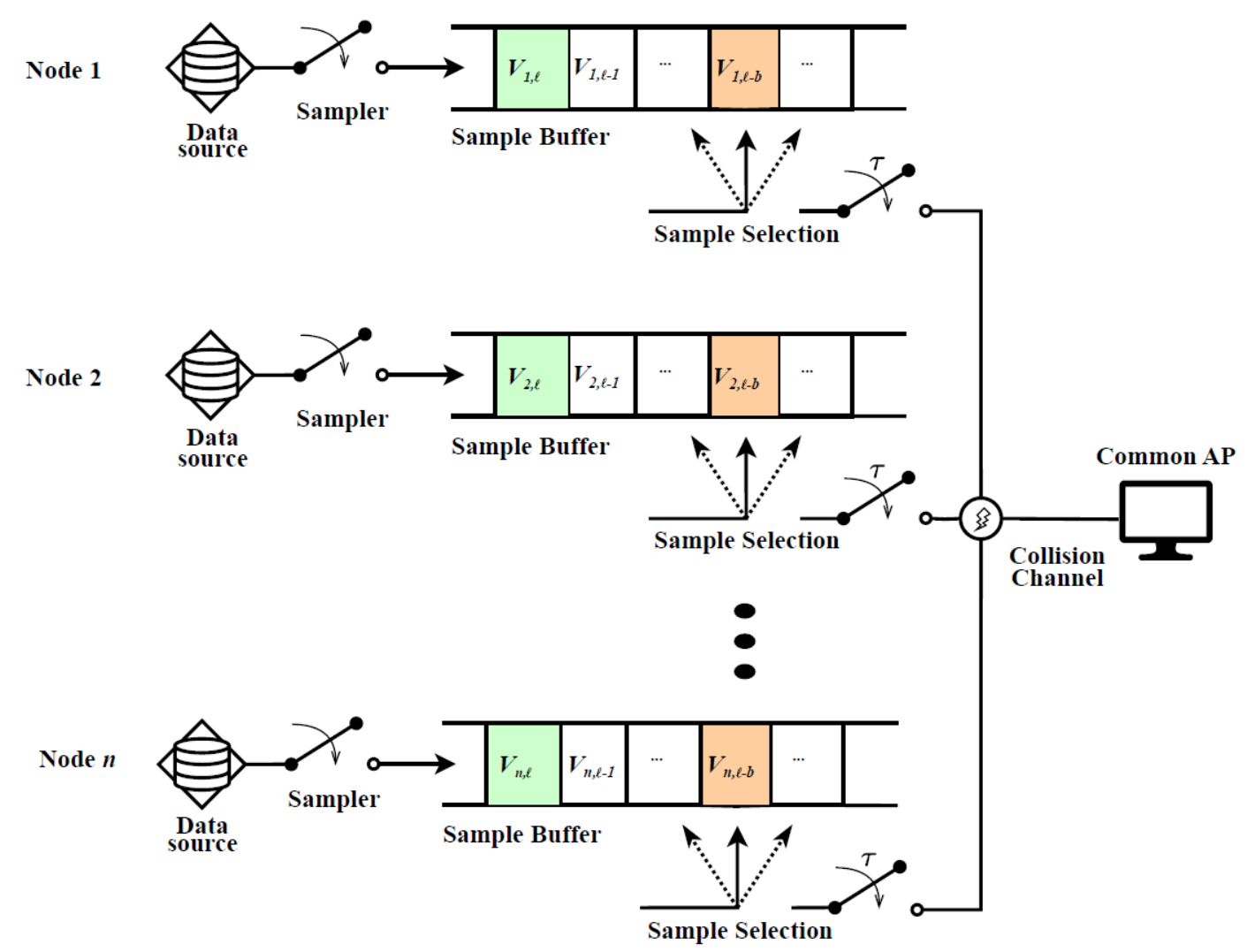}
    \caption{System Model.}\vspace{-0.5cm}
    \label{fig:system_model}
\end{figure}

\begin{itemize}
    \item We introduce Goal-Oriented Random Access (GORA), a novel random access strategy that leverages AoI as an auxiliary metric to ensure the timely transmission of packets that are most useful for the task at the destination.
    \item In GORA, we jointly optimize channel access and sample selection, adopting the \textit{selection-from-buffer} model proposed in \cite{shisher2024timely}.
    \item We perform a steady-state analysis of the proposed scheme, providing structural results on the optimal age threshold for a given penalty function, and the optimal age of the packet to be selected for transmission.
    \item We show that GORA outperforms Threshold ALOHA (TA) for non-monotone age penalty.
\end{itemize}
}

\vspace{-0.5em}
\section{System Model and Problem Formulation}
We consider the system depicted in Fig. \ref{fig:system_model}, where $n$ nodes access a common destination through a collision channel. {Time~is divided into equal-sized slots, each of duration $d$. The $\ell$-th time slot corresponds to the interval $\big[(\ell-1)d, \ell d\big)$.} Nodes are allowed to initiate transmissions only at the beginning of these time slots. When a single node transmits during a time slot, its packet is successfully received by the common destination at the end of the slot. However, if multiple nodes attempt transmission during the same time slot, their packets collide and are lost. Following each successful transmission, the transmitting node receives a 1-bit acknowledgment. \par
Each node~$i$ generates a distinct data flow~$i$. The Age of Information (AoI) $\Delta_i(t)$ of flow $i$ at time $t$ is defined as \cite{kaul2012real}
\begin{equation} \label{cont_AoI}
    \Delta_i(t) = t - u_i(t),
\end{equation} where $u_i(t)$ represents the generation time of the most recently received packet of flow $i$ at the common destination by time~$t$. \par
Each data flow~$i$ is associated with a goal function~$h_i(\delta)$, which characterizes the task of flow~$i$ at the common destination. The function $h_i(\delta)$ represents the penalty incurred by flow~$i$ when $\Delta_i(t) = \delta$ and is not necessarily a monotonic function of the AoI. To address this potential non-monotonicity, we adopt the selection-from-buffer model proposed in~\cite{shisher2024timely}. According to this model, each node~$i$ generates a packet, denoted by $V_{i,\ell}$, from the data source it monitors, $V_i(t)$, at the start of each time slot $\ell$ (i.e., at time instants $t = (\ell-1)d$ for $\ell \in \mathbb{Z}^+$). This packet is then appended to its buffer. Consequently, at any time slot $\ell$, node~$i$ has access to the packets $V_{i,\ell}, V_{i,\ell-1}, V_{i,\ell-2}, \ldots$. \par
We employ the continuous AoI $\Delta_i(t)$, defined in \eqref{cont_AoI}, for constructing the goal function~$h_i(\delta)$. {To leverage the time-slotted structure of our system in the analysis, we now define the discrete AoI~$\Delta_i[\ell]$ of flow~$i$, representing the number of time slots elapsed since the generation of the most recently received packet of flow~$i$ by time slot~$\ell$. The relationship between the continuous and discrete AoI is given by  
\begin{equation} \label{disc_AoI}
    \Delta_i[\ell] = \frac{\Delta_i( (\ell-1) d )}{d}.
\end{equation}} \vspace{-0.3cm} \par
We { adopt} a stationary packet selection and channel access strategy for each node~$i$. At the beginning of each time slot~$\ell$, if active, node~$i$ selects the packet $V_{i,\ell-b_i}$ with discrete AoI~$b_i$ from its buffer. This selected packet is transmitted in time slot $\ell$ with probability $\tau_i$.{Following a successful transmission, node~$i$ backs off for $\Gamma_i$ time slots before becoming active again and resuming transmissions~\cite{yavascan2021analysis}.} 
\par
We now define a staircase function, $H_i(t)$, representing the average penalty experienced by flow $i$ during the time slot containing $t$:
\begin{equation} \label{new_penalty_fnc}
H_i(t) = \frac{1}{d} \int_{\Delta_i[\ell]d }^{(\Delta_i[\ell]+1)d} h_i(\delta) \, d\delta, 
\quad \text{if} \quad (\ell-1)d \leq t < \ell d.
\end{equation} 
The objective is to minimize the long-term ensemble average expected penalty by jointly optimizing the parameters $b_i$, $\tau_i$, and $\Gamma_i$, where $i = 1, 2, \ldots, n$, which specify the packet selection and channel access strategy of all nodes:
\begin{align} \label{problem_formulation}
    \min_{(\bar{b}, \bar{\tau}, \bar{\Gamma})} \lim_{\ell\rightarrow\infty} \frac{1}{n} \sum_{i=1}^{n} \mathbb E_{(\bar{b}, \bar{\tau}, \bar{\Gamma})}[H_i(\ell d)]
\end{align} where $\bar{b} = \begin{bmatrix}b_1 & b_2 & \ldots & b_n\end{bmatrix}$, $\bar{\tau} = \begin{bmatrix}\tau_1 & \tau_2 & \ldots & \tau_n\end{bmatrix}$, and $\bar{\Gamma} = \begin{bmatrix}\Gamma_1 & \Gamma_2 & \ldots & \Gamma_n\end{bmatrix}$.\par 
In this paper, we consider a particular scenario where the goal functions of all flows are identical, i.e., $h_i(\delta) = h(\delta)$ for all $i$. For this scenario, the symmetry of the model guarantees the existence of at least one optimal solution that is symmetric across nodes. Hence, it suffices to focus on the case where $b_i = b$, $\tau_i = \tau$, and $\Gamma_i = \Gamma$ for all $i$. {Under such symmetry, the processes $\{H_i(t)\}$ are statistically identical. Consequently, problem \eqref{problem_formulation} simplifies, and we obtain the following problem~formulation.

\noindent{\textbf{Main Problem:}}}
\begin{align} \label{simp_problem_formulation}
    \min_{(b, \tau, \Gamma)} \lim_{\ell\rightarrow\infty} \mathbb E_{(b, \tau, \Gamma)}[H(\ell d)].
\end{align} where $H(t)$ represents any one of the statistically identical processes $\{H_i(t)\}$. 
\par
Solving \eqref{simp_problem_formulation} requires jointly determining the parameters $b^*$, $\tau^*$, and $\Gamma^*$ for an arbitrary, possibly non-monotonic, goal function $h(\delta)$ that is specific to the task at the destination. 

\vspace{-0.5em}
\section{Steady State Analysis} \label{section:Age}
In this section, we demonstrate that the system exhibits steady-state behavior under our symmetric stationary packet selection and channel access strategy. 
{Under this strategy, the discrete AoI $\Delta_i[\ell]$ of flow $i$ evolves as follows:
\begin{equation}\label{Disc_AoI_Evolution}
    \Delta_i[\ell] = 
    \begin{cases}
        b + 1, & \text{if node $i$ successfully} \\
        & \text{transmits at time slot $\ell-1$}, \\
        \Delta_i[\ell-1]+1, & \text{otherwise}.
    \end{cases}
\end{equation}
} \par
Furthermore, node $i$ is active, transmitting with probability $\tau$ at each time slot $\ell$, if its discrete AoI {$\Delta_i[\ell]$} is greater than or equal to $b+\Gamma$. Therefore, to construct a stochastic model for analyzing the transmission traffic of the system, it is sufficient to use the truncated version $A_i^{b,\Gamma}[\ell]$ of the discrete AoI process {$\Delta_i[\ell]$} corresponding to flow $i$, which is given by
\begin{equation}\label{eq:aoi_gamma}
    A^{b,\Gamma}_i[\ell] = 
    \begin{cases}
        b+1, & \text{if node $i$ successfully} \\
        & \text{transmits at time slot $\ell-1$}, \\
        b+\Gamma, & \text{if $A^{b,\Gamma}_i[\ell-1] = b+\Gamma$}, \\
        A^{b,\Gamma}_i[\ell-1]+1, & \text{otherwise}.
    \end{cases}
\end{equation} \par
The vector of truncated discrete AoI processes  
$$\mathbf{A}^{b,\Gamma}[\ell] = \begin{bmatrix}A^{b,\Gamma}_1[\ell] & A^{b,\Gamma}_2[\ell] & \ldots & A^{b,\Gamma}_n[\ell]\end{bmatrix}$$ 
is a finite-state Markov Chain (FSMC), with state space $\mathcal{S}^{b,\Gamma} = \{b+1, b+2, \ldots, b+\Gamma\}^n$. Due to the symmetry of the system, the successful transmission probability $p_s(\ell,\tau)$ is identical for all active nodes at time slot $\ell$ and is given by
\begin{equation}\label{succ_prob}
    p_s(\ell,\tau) = \tau(1-\tau)^{m[\ell]-1},
\end{equation}
where $m[\ell]$ denotes the number of active nodes at time slot $\ell$. The transition probability matrix $P^{b,\Gamma}$ of the FSMC $\mathbf{A}^{b,\Gamma}[\ell]$ can then be built based on the successful transmission probability defined in \eqref{succ_prob}.\par
An FSMC $\mathbf{A}^{\Gamma}[\ell]$ was constructed and the existence of a unique steady-state distribution was shown (Proposition~2 in \cite{atabay2020improving}) for Threshold ALOHA.\footnote{The FSMC $\mathbf{A}^{\Gamma}[\ell]$ is equivalent to $\mathbf{A}^{b, \Gamma}[\ell]$ when $b=0$.}\footnote{Note that the back-off period ($\Gamma$ in our notation) is denoted by $\theta$ in \cite{atabay2020improving}.} The following lemma establishes that the FSMCs $\mathbf{A}^{\Gamma}[\ell]$ and $\mathbf{A}^{b,\Gamma}[\ell]$ share the same unique steady-state distribution, where $b$ only introduces a shift in the state definition of $\mathbf{A}^{b,\Gamma}[\ell]$.
\begin{lem}
    The FSMC $\mathbf{A}^{b,\Gamma}[\ell]$ has a unique steady-state distribution, which matches the steady-state distribution of the FSMC $\mathbf{A}^{\Gamma}[\ell]$ defined in \cite{atabay2020improving}.
\end{lem}
\begin{proof}
    Let $\mathcal{S}^\Gamma = \{1, 2, \ldots, \Gamma\}^n$ and $P^\Gamma$ represent the state space and transition probability matrix of the FSMC $\mathbf{A}^\Gamma[\ell]$, respectively. We first construct a bijective mapping $f:\mathcal{S}^{b,\Gamma} \mapsto \mathcal{S}^\Gamma$ such~that
    $$f(\textbf{s})=\textbf{s}+\begin{bmatrix}-b & -b & \ldots & -b\end{bmatrix}.$$ \par
    The state transition dynamics of both Markov chains are characterized by the successful transmission probability $p_s(\ell,\tau)$, as defined in equation \eqref{succ_prob}. Notably, $p_s(\ell,\tau)$ depends on the number of active nodes $m[\ell]$ rather than the exact individual discrete AoI values of the nodes. Because any state vector $s \in \mathcal{S}^{b,\Gamma}$ and its counterpart $f(s) \in \mathcal{S}^\Gamma$ have the same number of active nodes\footnote{For FSMCs $\mathbf{A}^{\Gamma}[\ell]$ and $\mathbf{A}^{b,\Gamma}[\ell]$, a node is active at time slot $\ell$ if its truncated discrete AoI process is equal to $\Gamma$ and $b + \Gamma$, respectively.}, the transition probability matrices $P^{b,\Gamma}$ and $P^{\Gamma}$ are equivalent. Therefore, the solutions $\pi^{b,\Gamma}$ and $\pi^{\Gamma}$ to the global balance equations, 
    \begin{equation}
        \pi^{b,\Gamma} P^{b,\Gamma} = \pi^{b,\Gamma} \quad \text{and} \quad \pi^{\Gamma} P^{\Gamma} = \pi^{\Gamma},
    \end{equation}
    are the same. In other words, the steady-state probability $\pi_s$ of any state vector $s \in \mathcal{S}^{b,\Gamma}$ is identical to the steady-state probability $\pi_{f(s)}$ of its counterpart $f(s) \in \mathcal{S}^\Gamma$.
\end{proof}
{In the rest, we make the following assumption.
\begin{assume} \label{Assumption_1}
    The successful transmission probability $p_s(\ell, \tau)$, given by equation~\eqref{succ_prob}, remains constant in the steady state, i.e., $p_s(\ell, \tau) = p_s(\tau)$.
\end{assume}

This assumption, which has been shown to hold for TA in the large network limit \cite[Corollary~1]{yavascan2021analysis}, is commonly adopted in the related literature~\cite{kadota2021infocom, chen2020infocom, gopal2018infocom, bianchi2000jsac, kwak2005exponential}.}

\vspace{-0.5em}
\section{Optimal Policy and Numerical Results}
Under assumption \ref{Assumption_1}, the successful transmissions of any node $i$, following the symmetric stationary packet selection and channel access strategy, form an arithmetic renewal process. The i.i.d. inter-renewal times $X_i$ are given by ${X_i = (\Gamma + Y)d}$, where $Y$ is a geometric random variable with parameter $p_s(\tau)$. The penalty (or reward) function associated with the renewal process of node $i$ is given by $H_i(t)$ in equation \eqref{new_penalty_fnc}, where $h_i(\delta) = h(\delta)$. \par
The renewal reward process corresponding to each node $i$ is statistically identical. Consequently, in the following analysis, we focus on minimizing the steady-state average expected penalty associated with any of these $n$ renewal processes, which is equivalent to solving problem \eqref{simp_problem_formulation}. Accordingly, we omit the dependence of all variables on the flow index $i$. \par
Using the key renewal theorem for arithmetic processes \cite[Chapter~3.5]{gallager1995dsp}, the time-average expected penalty, $\lim_{\ell \rightarrow \infty} \mathbb{E}_{(b, \tau, \Gamma)} [H(\ell d)]$, under the symmetric stationary packet selection and channel access strategy is given by
\begin{equation} \label{key_renewal_theorem}
    \lim_{\ell \rightarrow \infty} \mathbb{E}_{(b, \tau, \Gamma)} [H(\ell d)] = \frac{\mathbb{E}[R_n]}{\mathbb{E}[X]},
\end{equation}
where $\mathbb{E}[R_n]$ and $\mathbb{E}[X]$ denote the expected penalty per renewal and the expected renewal time, respectively.

\par

{
To express~\eqref{key_renewal_theorem} in a more explicit form, we compute $\mathbb{E}[R_n]$ and $\mathbb{E}[X]$ separately. Each successful transmission resets the discrete AoI to $b+1$ and initiates a new renewal interval. This interval lasts for $\Gamma + Y$ time slots, during which the discrete AoI increases by~$1$ in each slot. The penalty incurred in the $k$-th slot of the interval is given by $\int_{(b+k)d}^{(b+k+1)d} h(\delta) \, d\delta$. Accordingly, the numerator can be written as
\[
\mathbb{E}[R_n] = \mathbb{E} \left[\int_{(b+1)d}^{(b+\Gamma+Y+1)d} h(\delta) \, d\delta\right].
\]
Moreover, the denominator is simply $\mathbb{E}[X] = (\Gamma + \mathbb{E}[Y])d$.
}

\par

{
Since $Y$ is a geometric random variable with success probability $p_s(\tau)$, its mean is $\frac{1}{p_s(\tau)}$, and its PMF is given by $p_Y(y) = p_s(\tau)(1 - p_s(\tau))^{y-1}$ for $y = 1,2,\ldots$. Therefore, we can express the right-hand side of equation~\eqref{key_renewal_theorem} as
\begin{align} \label{L_b_tau_Gamma}
    \sum_{y=1}^\infty \left(\int_{(b+1)d}^{(b+\Gamma+y+1)d} h(\delta) \, d\delta\right) \frac{p_s(\tau)(1 - p_s(\tau))^{y-1}}{(\Gamma + \frac{1}{p_s(\tau)})d}.
\end{align}
}

\par

{
Let $L(b, \tau, \Gamma)$ denote the time-average expected penalty expression in~\eqref{L_b_tau_Gamma}. Then, problem~\eqref{simp_problem_formulation} can be reformulated~as}
\begin{equation} \label{main_problem}
    \min_{\substack{b \in \mathbb Z_{\geq 0}, \ \tau \in (0,1), \ \Gamma \in \mathbb Z_{\geq 0}}} L(b, \tau, \Gamma)
\end{equation} \par

The parameters $b$, $\tau$, and $\Gamma$ are inherently discrete, making them challenging to optimize analytically. To gain insight into the solution, we define the continuous relaxation of the problem \eqref{main_problem}, allowing $b$ and $\Gamma$ to take real values. {
\begin{equation} \label{cont_relax}
    \min_{\substack{b \in \mathbb R_{\geq 0}, \ \tau \in (0,1), \ \Gamma \in \mathbb R_{\geq 0}}} L(b, \tau, \Gamma)
\end{equation}

\par

In Theorem \ref{Theorem_1} below, we establish two structural relationships between the optimal parameters $b^*$, $\Gamma^*$, and $\tau^*$, which will be used in solving (\ref{main_problem}). }

\begin{thm} \label{Theorem_1}
    For a given $\tau$, the optimal parameters $b_{\tau}^*$ and $\Gamma_{\tau}^*$ satisfy the following equations:
    \begin{equation} \label{eq_1_Theorem_1}
        h((b_{\tau}^* + 1)d) = \mathbb E \left[ h((b_{\tau}^* + \Gamma_{\tau}^* + Y + 1)d) \right],
    \end{equation}
    \begin{equation} \label{eq_2_Theorem_1}
        \frac{\mathbb E \left[\int_{(b_{\tau}^*+1)d}^{(b_{\tau}^*+\Gamma_{\tau}^*+Y+1)d} h(\delta) \, d\delta \right]}{\left( \Gamma_{\tau}^* + \mathbb E[Y] \right)d} = \mathbb E \left[ h((b_{\tau}^* + \Gamma_{\tau}^* + Y + 1)d) \right].
    \end{equation} 
    Here, $Y$ is a geometric random variable with parameter $p_s(\tau)$.
\end{thm}

\begin{proof}
{Consider the time-average expected penalty $L(b, \tau, \Gamma)$ for a given $\tau$.} By applying the Leibniz rule, we express
\begin{equation} \label{Th1_Pf_S1}
    \frac{\partial L}{\partial b} = \frac{ \mathbb E \left[ h((b + \Gamma + Y + 1)d) \right] - h((b+1)d)}{ \Gamma + \mathbb E[Y]},
\end{equation} and
\begin{equation} \label{Th1_Pf_S2}
    \frac{\partial L}{\partial \Gamma} = \frac{\mathbb E \left[ h((b + \Gamma + Y + 1)d) \right]}{ \Gamma + \mathbb E[Y] } - \frac{\mathbb E \left[ \int_{(b+1)d}^{(b+\Gamma+Y+1)d} h(\delta) \, d\delta \right]}{\left( \Gamma + \mathbb E[Y] \right)^2d}.
\end{equation}

For any extreme point of the function $L(b, \Gamma)$, including the global minimum, the right-hand sides of both equations \eqref{Th1_Pf_S1} and \eqref{Th1_Pf_S2} are equal to zero. Consequently, we obtain
\begin{equation}
    h((b+1)d) = \mathbb E \left[ h((b + \Gamma + Y + 1)d) \right]
\end{equation} and
\begin{equation}
    \frac{\mathbb E \left[\int_{(b+1)d}^{(b+\Gamma+Y+1)d} h(\delta) \, d\delta \right]}{\left( \Gamma + \mathbb E[Y] \right)d} = \mathbb E \left[ h((b + \Gamma + Y + 1)d) \right].
\end{equation} This completes the proof.
\end{proof}

\begin{figure}
    \centering
    \includegraphics[width=0.75\linewidth]{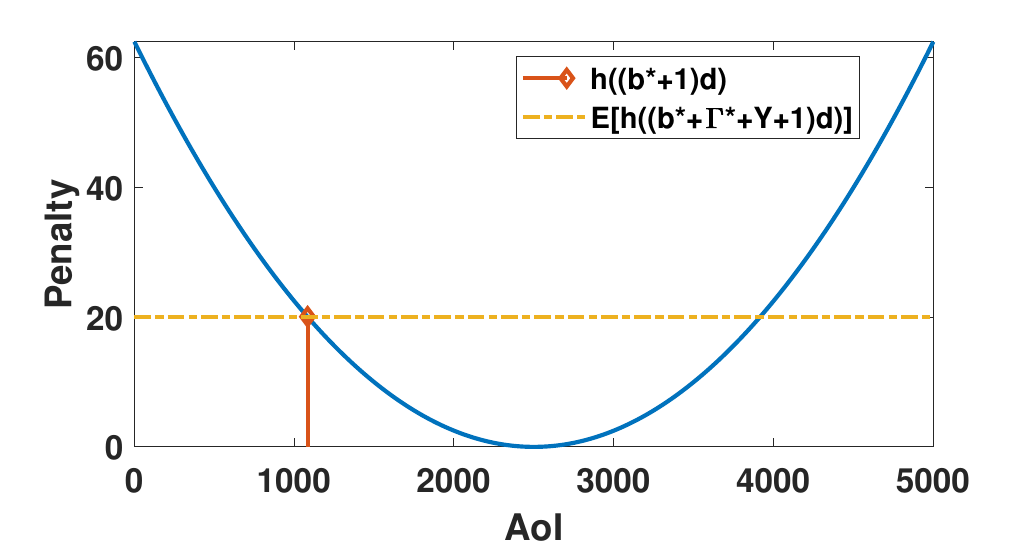}
    \caption{The optimal parameter $b^*$ and expected penalty at the end of a renewal interval, ${\mathbb E \left[ h((b^* + \Gamma^* + Y + 1)d) \right]}$, for $n=1000$ and a simple convex goal function.\vspace{-0.5cm}}
    \label{fig:b_eq_exp}
\end{figure}

Theorem \ref{Theorem_1} implies that both the penalty at the beginning of a renewal interval, $h((b^* + 1)d)$, and the optimal long-term average expected penalty are equal to the expected penalty at the end of a renewal interval, ${\mathbb E \left[ h((b^* + \Gamma^* + Y + 1)d) \right]}$. By setting an appropriate range for the $\tau$ values\footnote{The optimal $\tau$ is intuitively the reciprocal of the number of active nodes in the steady state. Under TA, the number of active nodes in the network converges in probability to $k_0n$ as $n$ increases, where $k_0$ is a constant~\cite{yavascan2021analysis}. This also applies to GORA, as the parameter $b^*$ does not influence the transmission strategies of the nodes.} and using Theorem \ref{Theorem_1}, the optimal parameters can be found. \par
Fig. \ref{fig:b_eq_exp} illustrates the optimal $b^*$ and the expected penalty at the end of a renewal interval, ${\mathbb E \left[ h((b^* + \Gamma^* + Y + 1)d) \right]}$, for $n=1000$ nodes and a simple convex goal function $h(\delta)$. This numerical result confirms equation~\eqref{eq_1_Theorem_1} of Theorem~\ref{Theorem_1}. \par
Note that there may be multiple $b$ and $\Gamma$ pairs satisfying \eqref{eq_1_Theorem_1} and \eqref{eq_2_Theorem_1} in Theorem \ref{Theorem_1} for a given $\tau$ if the time-average expected penalty function $L(b, \tau, \Gamma)$ is not convex. In such cases, the optimal pair, $b_{\tau}^*$ and $\Gamma_{\tau}^*$, should be determined by evaluating $L(b, \tau, \Gamma)$. In our numerical results, we verified convexity within the search region by computing the Hessian of the function $L(b, \tau, \Gamma)$. See Appendix \ref{Hessian} for details.\par

\begin{cor} \label{optimal_buffer_position} The optimal $b^*$ is not necessarily equal to zero for a non-monotone goal function $h(\delta)$.
\end{cor}

Corollary \eqref{optimal_buffer_position} states that transmitting the freshest packet in the buffer is not always optimal for nodes. Depending on the goal function~$h(\delta)$ and the network size~$n$, transmitting a staler packet may result in better task performance at the destination. Fig. \ref{fig:h1_b} illustrates the optimal parameter $b^*$ for varying network sizes $n$, considering a simple convex goal function. When the network is not overly crowded, nodes prefer to transmit a staler packet to maintain the AoI of the flow around values where the task-related penalty at the destination is small. However, as the network size increases, nodes tend to transmit fresher packets due to the reduced frequency of successful transmissions. 
\par
{Fig.~\ref{fig:h1_gora_vs_ta} illustrates the time-average penalty achieved by GORA, TA, and Slotted ALOHA (SA) for varying network sizes, considering the goal function depicted in Fig.~\ref{fig:h1_b}. Because minimizing the AoI does not always yield optimal performance due to the non-monotonicity of the goal function, by leveraging the selection-from-buffer data generation model \cite{shisher2024timely}, GORA outperforms both the conventional random access policy, SA, and the time-average AoI-minimizing policy, TA. However, when the network becomes overly crowded, i.e., when the expected time intervals between successive successful transmissions become excessively large, GORA also selects fresh packets, causing its performance advantage to disappear. Notably, GORA and TA yield identical results when $b^*=0$.
}

\par

\begin{figure}
    \centering
        \includegraphics[width=0.75\linewidth]{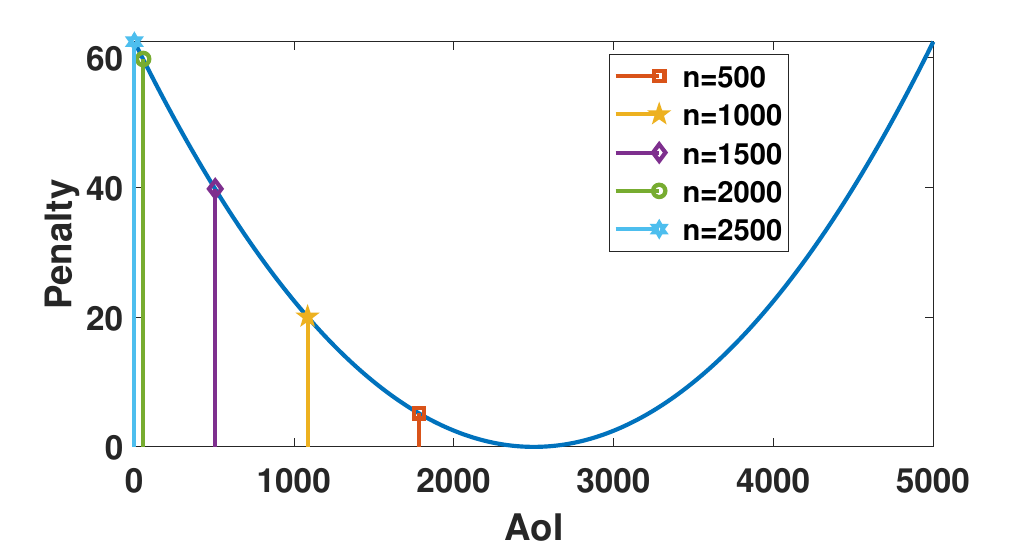}
        \caption{The optimal parameter $b^*$ for a simple convex goal function as the network size $n$ varies from 500 to 2500.\vspace{-0.5cm}}
        \label{fig:h1_b}
\end{figure}

\begin{figure}
    \centering
    \includegraphics[width=0.75\linewidth]{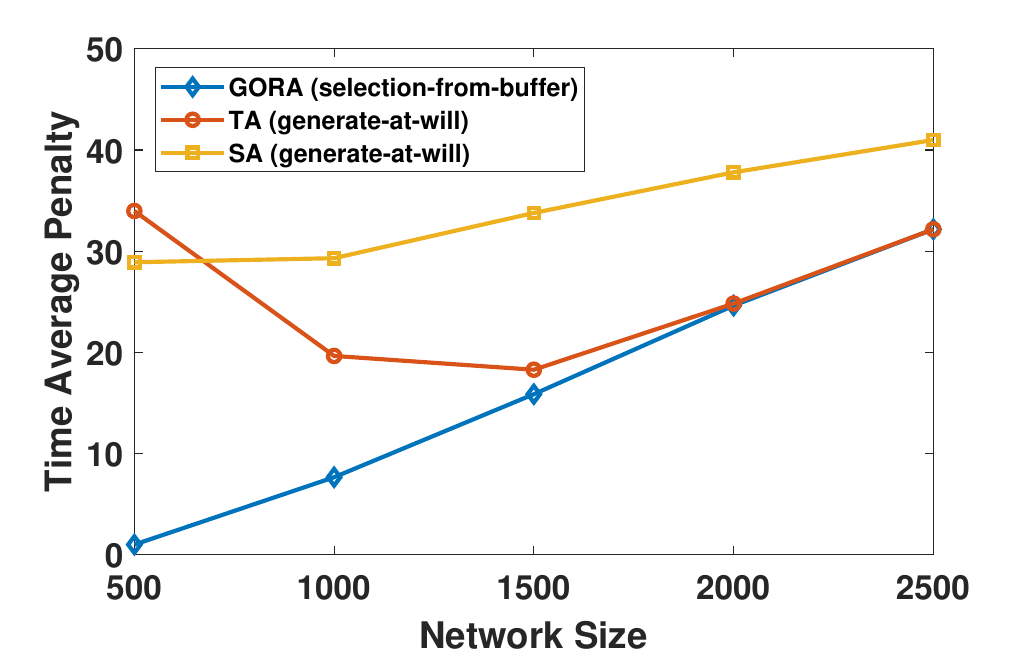}
    \caption{The performance evaluation of GORA, TA \cite{yavascan2021analysis, atabay2020improving} and SA for the goal function depicted in Fig. \ref{fig:h1_b}.\vspace{-0.5cm}}
    \label{fig:h1_gora_vs_ta}
\end{figure}

\begin{cor}\label{cor:global_min}
    Let $\mathcal{X}$ be the set of AoI values $\delta$ for which the goal function $h(\delta)$ is minimized. In the optimal setting, none of the values $x \in \mathcal{X}$ may satisfy
    \begin{equation}
        (b^*+1)d < x < (b^* + \Gamma^* + \mathbb{E}[Y] + 1)d
    \end{equation}
    for some non-monotone goal functions $h(\delta)$ and sufficiently large network size $n$.
\end{cor}

Corollary \ref{cor:global_min} highlights an interesting phenomenon arising from the relationship between network size and the optimal parameter $b^*$. Inherent to the random access concept, the rate of successful transmission opportunities for individual nodes decrease as the number of nodes in the system increases. Since the objective is to minimize the average penalty, the policy parameter $b$ cannot be greedily chosen to operate around an AoI value $\delta'$ where the goal function $h(\delta')$ is minimized. Depending on the expected duration between consecutive successful transmissions, operating at AoI values away from $\delta'$ may result in better task performance at the destination. \par
Fig. \ref{fig:h2_b} indicates the optimal parameter $b^*$ for varying network sizes $n$, considering a goal function with two minima. For network sizes $n$ between $500$ and $2000$, operating near the local minima results in better task performance at the destination. This behavior occurs as the penalty increases sharply for AoI values near the global minima, making it challenging to maintain a low steady-state average expected penalty within the intervals between consecutive successful transmissions. However, when the network size reaches $n = 2500$, the optimal parameter $b^*=0$, meaning GORA and TA become identical. 
{Fig. \ref{fig:h2_gora_vs_ta} presents the time-average penalty achieved by GORA, TA, and SA for the goal function depicted in Fig. \ref{fig:h2_b}.
Consistent with the results in Fig.~\ref{fig:h1_gora_vs_ta}, GORA outperforms the benchmark policies.
}
\par

\begin{figure}
    \centering
    \includegraphics[width=0.75\linewidth]{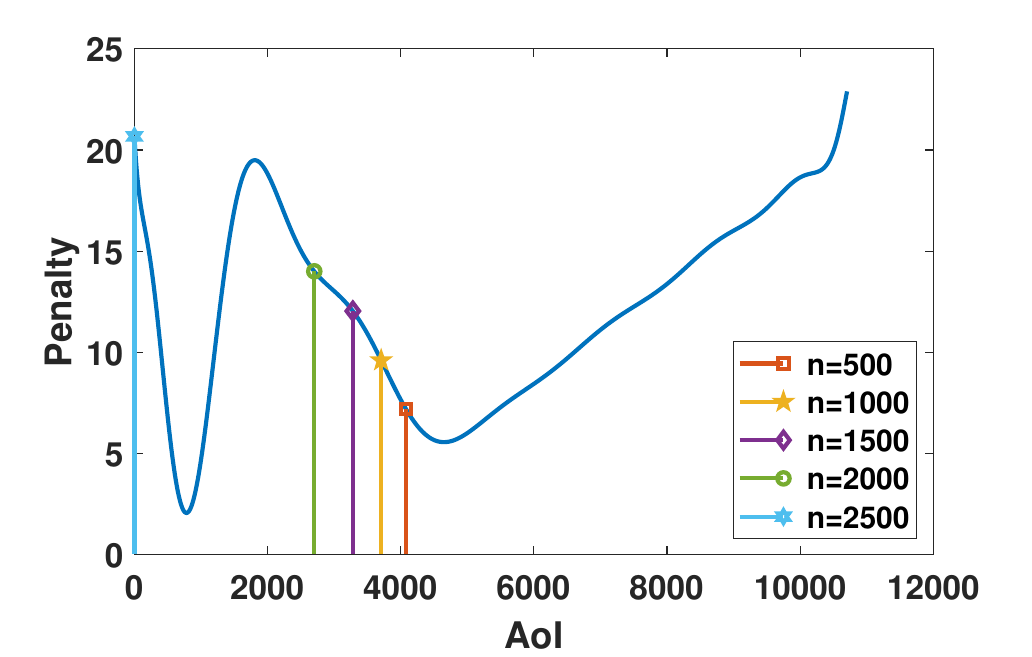}
    \caption{The optimal parameter $b^*$ for a goal function with two minima as the network size $n$ varies from 500 to 2500.\vspace{-0.5cm}}
    \label{fig:h2_b}
\end{figure}

\begin{figure}
    \centering
    \includegraphics[width=0.75\linewidth]{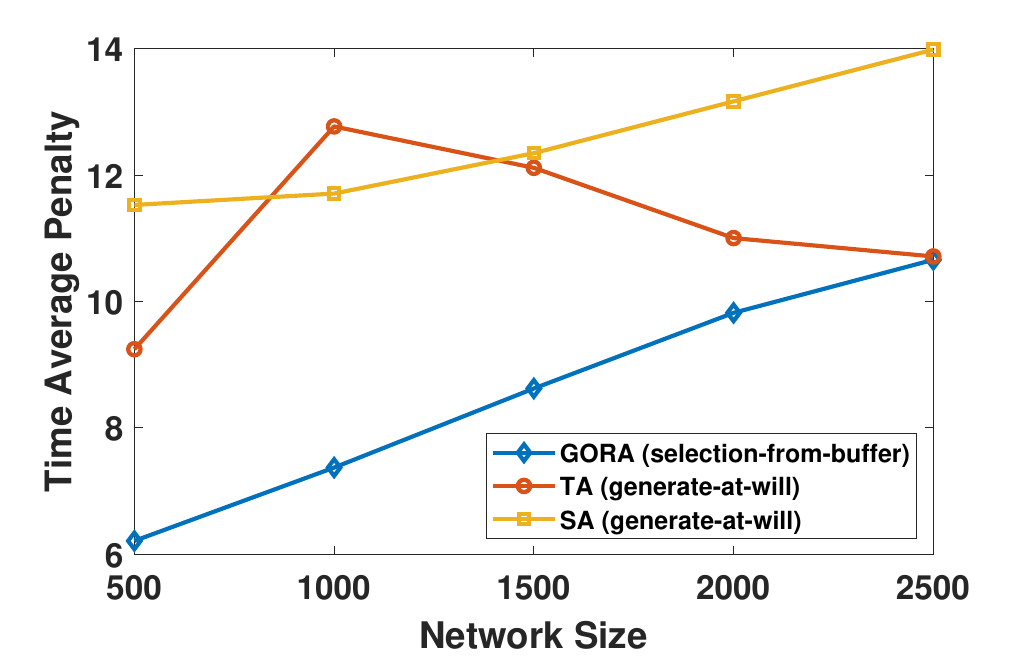}
    \caption{The performance evaluation of GORA, TA \cite{yavascan2021analysis, atabay2020improving} and SA for the goal function depicted in Fig. \ref{fig:h2_b}.\vspace{-0.5cm}}
    \label{fig:h2_gora_vs_ta}
\end{figure}

\vspace{-0.5em}
\section{Conclusion}
We introduced Goal-Oriented Random Access (GORA). This is a modification of conventional slotted ALOHA where transmission decisions are coupled with the selection of what to transmit, considering the ultimate goal of the data transmission. By leveraging the Age of Information (AoI) as an auxiliary metric, GORA allows nodes to strategically delay transmissions and send aged data samples, thereby enabling the effective accomplishment of the task at the destination. Our results demonstrate the superiority of GORA over age-aware random access policies for goal functions that are monotone or non-monotone in AoI. 
{Future work may explore systems that accommodate users with different, possibly non-monotone, goal functions and extend GORA by incorporating more advanced channel access strategies such as reservation-based or contention-based methods.}

\vspace{-0.5em}
\bibliographystyle{IEEEtran}
\bibliography{gora}

\vspace{-0.5em}
\appendices
\section{} \label{Hessian}
For ease of notation, let $\vartheta$, $\vartheta'$, and $\varphi$ represent the terms ${h((b + \Gamma + Y + 1)d)}$, ${h'((b + \Gamma + Y + 1)d)}$, and ${\Gamma + \mathbb{E}[Y]}$, respectively. The first entry and determinant of the Hessian matrix $\mathbb{H}$ of the function $L(b, \tau, \Gamma)$ for a given $\tau$ can then be expressed as follows:
\begin{equation} \label{first_entry}
    \frac{\partial^2 L}{\partial b^2} = \frac{ \left( \mathbb E \left[ \vartheta' \right] - h'((b+1)d) \right) d}{\varphi}
\end{equation}
\begin{align} \label{hessian}
    &\det(\mathbb{H}) =  \left( \frac{ \left( \mathbb E \left[ \vartheta' \right] - h'((b+1)d) \right) d}{\varphi} \right)\nonumber\\
    &\left( \frac{\varphi\mathbb E \left[ \vartheta' \right]d-2\mathbb E \left[ \vartheta \right]}{\varphi^2} + \frac{2\mathbb E \left[ \int_{(b+1)d}^{(b+\Gamma+Y+1)d} h(\delta) d\delta\right]}{\varphi^3d}\right) - \nonumber\\
    &\left( \frac{\varphi\mathbb E \left[ \vartheta' \right]d- \mathbb E \left[ \vartheta \right]+h((b+1)d)}{\varphi^2} \right)^2
\end{align}
If the Hessian matrix $\mathbb{H}$ is positive definite, then the function $L(b, \tau, \Gamma)$ for a given $\tau$ is convex. A necessary and sufficient condition for the $2 \times 2$ Hessian matrix $\mathbb{H}$ to be positive definite is that both \eqref{first_entry} and \eqref{hessian} are strictly positive.

\vfill

\end{document}